\documentclass[envcountsame,runningheads]{llncs}

\usepackage{graphicx}
\usepackage{amsmath}
\usepackage{amssymb}
\usepackage{amsfonts}
\usepackage{mathtools}
\usepackage{paralist}
\usepackage{xspace}
\usepackage{url}
\usepackage{comment}
\usepackage[ruled,linesnumbered,noend]{algorithm2e}

\usepackage{tikz}
\usetikzlibrary{arrows,shapes,snakes,automata,backgrounds,petri}
\usetikzlibrary{positioning}

\newcommand{\tup}[1]{\langle #1 \rangle}
\newcommand{\cali}[1]{\mathcal{#1}}

\newcommand{\false}{\textit{false}}
\newcommand{\past}{\textit{past}}

\newcommand{\lca}{\textit{LCA}}
\newcommand{\order}{\textit{Order}}

\DeclarePairedDelimiter\floor{\lfloor}{\rfloor}

\usepackage{color}
\newcount\Comments  
\newcommand{\mytodo}[2]
{\ifnum\Comments=1
	{\marginpar{\small{\color{#1}	#2}}}\fi}

\makeatletter
\newcommand{\nosemic}{\renewcommand{\@endalgocfline}{\relax}}
\newcommand{\dosemic}{\renewcommand{\@endalgocfline}{\algocf@endline}}
\let\oldnl\nl
\newcommand{\nonl}{\renewcommand{\nl}{\let\nl\oldnl}}
\makeatother

\begin{document}

\title{Better Incentives for Proof-of-Work}

\author{Jakub Sliwinski, Roger Wattenhofer}
\email{}
\institute{ETH Zurich \\
\email{{jsliwinski,wattenhofer}@ethz.ch}}


\maketitle

\begin{abstract}
This work proposes a novel proof-of-work blockchain incentive scheme such that, barring exogenous motivations, following the protocol is guaranteed to be the optimal strategy for miners. Our blockchain takes the form of a directed acyclic graph, resulting in improvements with respect to throughput and speed.

More importantly, for our blockchain to function, it is not expected that the miners conform to some presupposed protocol in the interest of the system's operability. Instead, our system works if miners act selfishly, trying to get the maximum possible rewards, with no consideration for the overall health of the blockchain.
\end{abstract}

\section{Introduction}\label{sec:intro}

A decade ago, Satoshi Nakamoto presented his now famous Bitcoin protocol \cite{nakamoto2008bitcoin}. Nakamoto assembled some stimulating techniques in an attractive package, such that the result was more than just the sum of its parts.

The Bitcoin blockchain promises to order and store transactions meticulously, despite being anarchistic, without a trusted party. Literally anybody can participate, as long as ``honest nodes collectively control more CPU power than any cooperating group of attacker nodes." \cite{nakamoto2008bitcoin}

In Section 6 of his seminal paper, Nakamoto argues that it is rational to be honest thanks to block rewards and fees. However, it turns out that Nakamoto was wrong, and rational does not imply honest. If a miner has a fast network \textit{and/or} a significant fraction of the hashing power, the miner may be better off by not being honest, holding blocks back instead of immediately broadcasting them to the network \cite{eyal2014selfish}.

If the material costs and payoffs of mining are low, one can argue that the majority of miners will want to remain honest. After all, if too many miners stop conforming to the protocol, the system will break down. However, the costs and payoffs of participation vary over time, and a majority of miners remaining altruistic is never guaranteed. Strategies outperforming the protocol may or may not be discovered for different blockchain incentive designs. However, as long as it is not proven that no such sophisticated strategy exists, the system remains in jeopardy.

\subsection{Blockchain Game}\label{sec:basics}

Typical blockchains, such as Bitcoin's, take the form of a rooted tree of blocks. During the execution of the protocol, players continually create new blocks that are appended to the tree as new leaves. Creating blocks is computationally intensive, so that the network creates a specific number of blocks in a given time period, such as one block every ten minutes on average in Bitcoin. One path of blocks, such as the longest path, is distinguished as the main chain and keeps being extended by addition of new leaves. The network's participants want to create blocks that remain incorporated into the main chain, as these blocks are rewarded. Ideally, the leaves would be added in sequence, each leaf appended to the previous leaf. However, by chance or malice, it is inevitable that some leaves are appended to the same block and create a ``fork". Then, it is uncertain which one will end up extending the main chain. According to typical solutions, one of the competing leaves is eventually chosen as being in the main chain, and the creator of the other leaf misses out on block rewards. This approach introduces some unwanted incentives and a potential to punish other players. Even worse, some factors such as network connectivity start to play a role and might influence the behaviour of players.

\subsection{Our Contribution}\label{sec:contrib}

We propose a blockchain design with an incentive scheme guaranteeing that deviating from the protocol strictly reduces the overall share and amount of rewards. All players following the protocol constitute a strict, strong Nash equilibrium. Our approach is to ensure that creating a fork will always be detrimental to all parties involved. Our design allows blocks to reference more than one previous block; in other words, the blocks form a directed acyclic graph (DAG). We prove that miners creating a new block have an incentive to always reference all previously unreferenced blocks. Hence, all blocks are recorded in the blockchain and no blocks are discarded.

\subsection{Intuitive Overview}

In Section \ref{sec:model} we describe the terms to define our protocol.

In Section \ref{sec:dag} we explain the protocol and how to interpret the created DAG. In terms of security, our design is identical to known proof-of-work blockchains, as similarly to other protocols, we identify the main chain to achieve consensus. Intuitively, each new block should reference all previous terminal blocks known to the miner and automatically extend the main chain. In Subsection \ref{sec:order}, we explain how to use the main chain to process and totally order all blocks \cite{lewenberg2015inclusive}.

In Section \ref{sec:rewards} we construct and discuss our reward scheme.

In Subsection \ref{sec:stale} we explain how to label some blocks as stale, such that blocks mined by honest miners are not labeled as stale, but blocks withheld for a long time are labeled as stale. Stale blocks do not receive any rewards.

The core idea of the incentive scheme is to penalize every block by a small amount for every block that it ``competes" with.

In Section \ref{sec:related} we discuss related work.

\section{Model and Preliminaries}\label{sec:model}

\subsection{Rounds}

We assume a network with a message diffusion mechanism that delivers messages to all connected parties (similarly to Bitcoin's network).

Similarly to foundational works in the area \cite{garay2015bitcoin} we express the network delay in terms of rounds. Communication is divided into rounds, such that when a player broadcasts a message, it will be delivered to all parties in the network in the next round. Thus each round can be viewed as: 1) receiving messages sent in the previous round, 2) computing (mining) new blocks, 3) broadcasting newly found blocks to all other players.

Rounds model the network delay for the purpose of analysis. However, the protocol itself is not concerned with the division of time into rounds in any way, and only relies on the network delay being correspondingly bounded.

\subsection{Players}

To avoid confusion in how we build on previous work, we stick to the usual terminology of {\it honest} players and an {\it adversary}. The players that conform to the protocol are called honest. A coalition of all parties that considers deviating from the protocol is controlled by an adversary. We gradually introduce new elements, and eventually show that by deviating from the protocol, the adversary reduces its share and amount of rewards. Hence, rational becomes synonymous with honest.

The adversary constitutes a minority as described in Section \ref{sec:mining}, otherwise the adversary can take over the blockchain by simply ignoring all actions by honest players.

The adversary is also more powerful than honest players. First of all, we consider the adversary as a single entity. The adversary does not have to send messages to itself, so the mine/send/receive order within a round does not apply to the adversary. Moreover, the adversary gets to see all messages sent by honest players in round $r$ before deciding its strategy of round $r$. After seeing the honest messages, the adversary is not allowed to create new blocks again in this round. Moreover, the adversary controls the order that messages arrive to each player.

\subsection{Blocks}

Blocks are the messages that the players exchange, and a basic unit of the blockchain. Formally, a block $B$ is a tuple $B = \tup{\cali T_B, \cali R_B, c, \eta}$, where: 
\begin{itemize}
    \item $\cali T_B$ is the content of the block
    \item $\cali R_B$ is a set of references (hashes) to previously existing blocks, i.e. $\cali R_B =  \{h(B_1),\dots,h(B_m)\}$
    \item  $c$ is a public key of the player that created the block
    \item $\eta$ is the proof-of-work nonce, i.e., a number such that for a hash function $h$ and difficulty parameter $D$, $h(B) < D$ holds.
\end{itemize}

The content of the block $\cali T_B$ depends on the application. In general, $\cali T_B$ contains some information that the block creator wishes to record in the blockchain for all participants to see. We consider blockchain properties independently of the content $\cali T_B$. The content $\cali T_B$ is discussed in Section \ref{sec:blockcontents} of the appendix.

The creator of $B$ holds the private key corresponding to $c$. The creator can later use the key to withdraw the reward for creating $B$. The amount of reward is automatically determined by the protocol, and at the core of our contribution in Section \ref{sec:rewards}.

\subsection{DAG}

$\cali R_B$ includes at least one hash of a previous block, which might be the hash of a special {\it genesis} block $\tup{\emptyset, \emptyset, \bot, 0}$. The hash function is pre-image resistant, i.e. it is infeasible to find a message given its hash. If a block $B'$ includes a reference to another block $B$, $B'$ must include $h(B)$, and hence has to be created after $B$.

A directed cycle of blocks is impossible, as the block which was created earliest in such a cycle cannot include a hash to the other blocks that were created later. Consequently, the blocks always form a directed acyclic graph ({\it DAG}) with the genesis block as the only root (block without any parent) of this DAG.

\subsection{Mining}\label{sec:mining}

Creating a new block is achieved by varying $\eta$ to find a hash value that is smaller than the difficulty parameter $\cali D$, i.e., $h(\tup{\cali T_B, \cali R_B, c, \eta}) < \cali D$. Creating blocks in this way is called {\it mining}. Blocks are called honest if mined by an honest player, or adversarial if mined by the adversary.

By varying $\cali D$, the protocol designer can set the probability of mining a block with a single hashing query arbitrarily. The difficulty $\cali D$ could also change during the execution of the protocol to adjust the rate at which blocks are created. We leave the details of changing $\cali D$ to future work, and assume $\cali D$ to be constant.

The honest players control the computational power to mine $\alpha$ blocks in expectation in one round. The computational power of the adversary is such that the expected number of blocks the adversary can mine in one round is equal to $\beta$. The adversary does not experience a delay in communication with itself, so the adversary might mine multiple blocks forming a chain in one round.

\paragraph{Assumptions} The following assumptions are made in order to satisfy the prerequisites of Lemma \ref{lem:freshblock}, which was proven in \cite{kiayias2017trees}. Lemma \ref{lem:freshblock} links our work to traditional blockchains. Intuitively, the lemma states that a traditional blockchain works with respect to the most basic requirement. If one believes a blockchain to function in this basic way under some other assumptions, those assumptions can be used instead, and our results would apply in the same way.

Because of the delay in communication, the effective computational power of the honest players corresponds to the probability $\alpha' \approx \alpha e^{-\alpha}$ \cite{kiayias2017trees} that in a given round exactly one honest player mines a block.

\begin{enumerate}
\item The honest players have more mining power: $\alpha' \geq \beta (1 + \epsilon)$ for a constant $\epsilon > 0$.

\item The difficulty $D$ is set such that the expected number of blocks mined within one round is less than one: $\alpha + \beta < 1$.
\end{enumerate}

\subsection{Action Space}

The state of the blockchain is only updated through discovery and broadcasting of new blocks, hence the adversary can only vary its behaviour with respect to the following factors:

\begin{itemize}
    \item the blocks being mined i.e. the contents, the included references etc.
    \item when to announce any of the mined blocks
    \item the set of agents to whom to send a given block.\footnote{Honest agents disseminate all received blocks, so by sending a block to a subset of agents, the adversary can delay other agents from seeing a block for only one round.}
\end{itemize}

\section{The Block DAG}\label{sec:dag}

The protocol by which the honest players construct the block DAG is simple:
\begin{itemize}
\item Attempt to mine new blocks.
\item Reference in $\cali R_B$ all unreferenced blocks observed.
\item Broadcast newly mined blocks immediately.
\end{itemize}

Each player stores the DAG formed by all blocks known to the player. For each block $B$, one of the referenced blocks $B_i$ is the parent $B_i = P(B)$, and $B$ is the child of $P(B)$. The parent is automatically determined based on the DAG structure. The parent-child edges induce the {\it parent tree} from the DAG.

The players use Algorithm \ref{algo:ghost} by \cite{sompolinsky2015secure} to select a chain of blocks going from the genesis block to a leaf in the parent tree. The selected chain represents the current state of the blockchain; it is called the {\it main} chain. The main chain of a player changes from round to round. Players adopt main chains that may be different from each other, depending on the blocks observed.

\begin{algorithm}
\DontPrintSemicolon
\KwIn{a block tree $T$}
\KwOut{block $B$ - the end of the selected chain}
$B \gets \textit{genesis}$ \tcp*[r]{start at the genesis block.}
\While{$B \textit{ has a child in } T$}{
$B \gets \textit{ heaviest child of } B$
\tcp*[r]{continue with the child of B}
\tcp*[r]{with most nodes in its subtree.}
}
\Return $B$

\caption{Main chain selection algorithm.}
\label{algo:ghost}
\end{algorithm}

Let $\past(B)$ denote the set of blocks reachable by references from $B$ and the DAG formed by those blocks. The protocol dictates referencing all blocks that otherwise would not be included in $\past(B)$. Then, by creating a new block $B$, the creator communicates only being aware of blocks in $\past(B)$. Based on $\past(B)$, we determine $P(B)$ as the end of the main chain (Algorithm \ref{algo:ghost}) of the DAG of the player when creating a new block $B$ \cite{lewenberg2015inclusive}.

\begin{definition}[Determining Parent]\label{def:parent}
For a given block $B$, the block returned by Algorithm \ref{algo:ghost} in the parent tree of $\past(B) \setminus \{B\}$ is the parent of $B$.
\end{definition}

Lemma \ref{lem:freshblock} by \cite{kiayias2017trees}, encapsulates the notion that a blockchain (represented by the parent tree in our description) functions properly with respect to a basic requirement. Intuitively, it states that from any point in time, the longer one waits, the more probable it becomes that some honest block mined after that point in time is contained in a main chain of each honest player. The probability of the contrary decreases exponentially with time.

\begin{lemma}[Fresh Block Lemma]\label{lem:freshblock}
For all $r,\Delta \in \mathbb{N}$, with probability $1 - e^{-\Omega(\Delta)}$, there exists a block mined by an honest player on or after round $r$ that is contained in the main chain of each honest player on and after round $r+\Delta$.
\end{lemma}

Lemma \ref{lem:freshblock} can be proved with respect to other chain selection rules, for instance picking the child with the longest chain instead of the heaviest child as in Algorithm \ref{algo:ghost}. Our work can be applied equally well using such chain selection rules.

If the protocol designer has control over some factor $x$, probability of the form $e^{-\Omega(x)}$ can be set arbitrarily low with relatively small variation of $x$. Probability of the form $e^{-\Omega(x)}$ is called negligible.\footnote{Probabilities of this form are often disregarded completely in proofs \cite{pass2017fruitchains}.}

\subsection{Block Order}\label{sec:order}

We will now explain, how all blocks reachable by references will be ordered, following the algorithm of \cite{lewenberg2015inclusive}. According to the resulting order, the contents of blocks that fall outside of the main chain can be processed, as if all blocks formed one chain.

\begin{definition}\label{def:order}
Each player processes blocks in the order $\order(B)$, where $B$ is the last block of the main chain.
\end{definition}

\begin{algorithm}
\DontPrintSemicolon
\KwIn{a block $B$}
\KwOut{a total order of all blocks in $\past(B)$}
On the first invocation, $\textit{visited}(\cdot)$ is initialized to $\false$ for each block.

\lIf{$\textit{visited}(B)$}{\Return $\emptyset$}
$\textit{visited}(B) \gets \textit{true}$ \tcp*[r]{Blocks are visited depth-first.}
\lIf{$B = \text{genesis}$}{
\Return $(B)$
}
$O \gets \order(P(B))$
\tcp*[r]{Get the order of $P(B)$ recursively.}
\For{$i = 1,\dots,m$}{
$O \gets O.\textit{append}(\order(B_i))$
\tcp*[r]{Append newly included blocks.}
}
$O \gets O.\textit{append}(B)$ \tcp*[r]{Append $B$ at the end.}
\Return $O$\;

\caption{$\order(B)$: a total order of blocks in $\past(B)$.}
\label{algo:inclusive}
\end{algorithm}

Note the order of executing the FOR loop in line 6 of the Algorithm \ref{algo:inclusive} has to be the same for each player for them to receive consistent orders of blocks. Algorithm \ref{algo:inclusive} processes $B_i$'s in the order of inclusion in $\cali R_B$, but the order could be alphabetical or induced by the chain selection rule.

Based on lines numbered 5-8 we can state Corollary \ref{cor:incorder}.

\begin{corollary}\label{cor:incorder}
$\order(B)$ extends $\order(P(B))$ by appending all newly reachable blocks not included yet in $\order(P(B))$.
\end{corollary}

\begin{lemma}\label{lem:allincluded}
Any announced block becomes referenced by a block contained in the main chain of any honest player after $\Delta$ rounds with probability $1 - e^{-\Omega(\Delta)}$.
\end{lemma}
\begin{proof}
Suppose a block $B$ is announced at round $r$. By Lemma \ref{lem:freshblock}, some honest block $A$ mined in the following $\Delta$ rounds is contained in the main chains adopted by honest players after round $r+\Delta$. Since $A$ is honest, $B \in \past(A)$.
\qed \end{proof}

By Lemma \ref{lem:allincluded} all announced blocks are eventually referenced in the main chains of honest players. Since for the purpose of achieving consensus we rely on the results of \cite{kiayias2017trees} and \cite{lewenberg2015inclusive}, we state Corollary \ref{cor:consensus}.

\begin{corollary}\label{cor:consensus}
The protocol achieves consensus properties corresponding to \cite{kiayias2017trees} and \cite{lewenberg2015inclusive}.
\end{corollary}

\section{Reward Schemes}\label{sec:rewards}

\subsection{Stale Blocks}\label{sec:stale}

We now introduce a mechanism to distinguish blocks that were announced within a reasonable number of rounds from blocks that were withheld by the miner for an extended period of time. Such withheld blocks are called {\it stale}. Honest miners broadcast their blocks immediately, so stale blocks can be attributed to the adversary. In our incentive scheme, stale blocks will not receive any rewards and will also be ignored for the purpose of determining other block rewards. Thus we ensure that it is pointless for the adversary to wait too long before broadcasting its blocks.

The basic definition of whether a block $A$ is stale is termed with respect to some other block $B$. We are only interested in blocks $B$ that form the main chain. When the main chain is extended, the sets of stale and non-stale blocks are preserved (and extended). Hence, stale-ness is determined by the eventual main chain.

\begin{definition}
Given a block $B$, the set of stale blocks $S_B$ is computed by Algorithm \ref{algo:stale}. Then, $\bar S_B = \past(B) \setminus S_B$. If $A \in S_B$ we call $A$ {\it stale}.
\end{definition}

The constant $p$ of Algorithm \ref{algo:stale} is chosen by the protocol designer. Intuitively, given a main chain ending with block $B$ that references another block $A$, we judge $A$ by the distance one needs to backtrack along the main chain to find an ancestor of $A$. If the distance exceeds $p$, $A$ is stale.

We call $P^i(B)$ the {\it $i^{\textit{th}}$ ancestor of $B$} and $B$ is a {\it descendant} of $P^i(B)$.\footnote{Note that ancestors and descendants are defined based on the parent tree and not based on other non-parent references building up the DAG.} By $\lca(B_1, B_2)$ (lowest common ancestor) we denote the block that is an ancestor of $B_1$ and an ancestor of $B_2$, such that none of its children are simultaneously an ancestor of $B_1$ and an ancestor of $B_2$.

For blocks $A$ and $B$, $D(A,B)$ is the distance between $A$ and $B$ in the parent tree, i.e. $D(A,P(A)) = 1$, $D(A,P(P(A))) = 2$, etc.

\begin{algorithm}
\DontPrintSemicolon
\KwIn{a block $B$}
\KwOut{a set $S_B$}

\lIf{$B = \textit{genesis}$}{
\Return $\emptyset$
}

$S \gets S_{P(B)}$ \tcp*[r]{Copy $S_{P(B)}$ for blocks in $\past(P(B))$.}
\For{$A \in \past(B) \setminus \past(P(B))$}{
$X = \lca(A, B)$\;
$\textit{Age} = D(X, B)$ \tcp*[r]{age = distance from $B$ to LCA.}
\If{$\textit{Age} > p$}{
$S = S \cup \{A\}$
\tcp*[r]{$A$ is stale iff age is bigger than $p$}
}
}
\Return $S$\;

\caption{Compute $S_B$.}
\label{algo:stale}
\end{algorithm}

Corollary \ref{cor:staleparent} shows that when the main chain is extended, the stale-ness of previously seen blocks is preserved.

\begin{corollary}\label{cor:staleparent}
If $A \in \past(P(B))$ then $A \in S_B \iff A \in S_{P(B)}$.
\end{corollary}
\begin{proof}
Line 2 in Algorithm \ref{algo:stale} sets $S_B$ as the same as $S_{P(B)}$, while the following FOR loop adds only blocks $A \notin \past(P(B))$.
\qed \end{proof}

Theorem \ref{thm:honestarenonstale} establishes the most important property of stale-ness. The probability that the adversary can successfully make an honest block stale decreases exponentially with $p$, and is negligible.

\begin{theorem}[Honest Blocks are Not Stale]\label{thm:honestarenonstale}
Let $B$ be an honest block mined on round $r$. With probability $1 - e^{-\Omega(p)}$, after round $r + O(p)$ each honest player $H$ adopts a main chain ending with a block $B_H$ such that $B \in \bar S_{B_H}$.
\end{theorem}

The proof is deferred to the appendix.

\subsection{Discussion of Flat Rewards}

Consider coupling the presented protocol with a reward mechanism $\cali R^0$ that, intuitively speaking, grants some flat amount $b$ of reward to all non-stale blocks, and $0$ reward to stale blocks. $\cali R^0$ is a special case of the reward scheme properly defined in Definition \ref{def:rewards}.

\begin{corollary}
Under the reward scheme $\cali R^0$, honest players are rewarded proportionally to the number of blocks they mine, except with negligible probability.
\end{corollary}
\begin{proof}
By Theorem \ref{thm:honestarenonstale} honest blocks are not stale, so honest miners receive rewards linear in the number of blocks they mined. The adversary might only decrease its rewards by producing stale blocks, otherwise the adversary is rewarded in the same way.
\qed \end{proof}

Note that $\cali R^0$ achieves the same fairness guarantee as the Fruitchains protocol to be discussed in Section \ref{sec:fruitchain} --- honest blocks are incorporated into the blockchain as non-stale, while withholding a block for too long makes it lose its reward potential. Both protocols rely on the honest majority of participants to guarantee this fairness.

The Fruitchains protocol relies critically on merged-mining \cite{mergepost} (also called 2-for-1 POW \cite{garay2015bitcoin}) fruits and blocks. While fruits are mined for the rewards, blocks are supposed to be mined entirely voluntarily with negligible extra cost. The reward scheme $\cali R^0$ avoids this complication.

Granting flat amount of reward for each non-stale block leaves a lot of room for deviation that goes unpunished. In the case of the Fruitchains protocol, mining blocks does not contribute rewards in any way. Hence, any deviation with respect to mining blocks (which decide the order of contents) is free of any cost for the adversary. In the context of cryptocurrency transactions, a rational adversary should always attempt to double-spend.

In the case of $\cali R^0$, the adversary can refrain from referencing some recent blocks, and suffer no penalty. However, attempting to manipulate the order of older blocks would render the adversary's new block stale, and hence penalize. Thus, we view even the base case $\cali R^0$ of the presented reward scheme as a strict improvement over the Fruitchains protocol.

\subsection{Penalizing Deviations}

Central to our design is the approach to treating forks i.e. blocks that ``compete'' by referencing the same parent block and not each other. Typically, blockchain schemes specify that one of the blocks eventually ``loses'' and the creator misses out on some rewards, essentially discouraging competition. However, there are ways of manipulating this process to one's advantage, and the uncertainty of which block will win the competition introduces unneeded incentives. We penalize all parties involved in creating a fork.

The {\it conflict set} introduced in Definition \ref{def:conflictset} contains the blocks that ``compete" with a given block. Stale blocks are excluded, as we ignore them for the purpose of computing rewards. Like stale-ness, the conflict set is defined with respect to some other block $A$. Again, we are only interested in blocks $A$ that form the main chain, and the conflict set indicated by the eventual main chain.

The conflict set of a non-stale block $B$ contains all non-stale blocks $X$ that are not reachable by references from $B$, and $B$ is not reachable by references from $X$.

\begin{definition}[Conflict Set]\label{def:conflictset}
For blocks $A$ and $B$ where $B \in \bar S_A$,
$$X_A(B) = \{X : X \in \bar S_A \land X \notin \past(B) \land B \notin \past(X)\}.$$
\end{definition}

\tikzstyle{block} = [draw=black, fill=gray!5, thick, minimum size=7mm,
    rectangle, rounded corners, inner sep=3pt, inner ysep=2pt, text centered]
\tikzstyle{block_blue} = [draw=black, fill=blue!20, thick, minimum size=7mm,
    rectangle, rounded corners, inner sep=3pt, inner ysep=2pt, text centered]
\tikzstyle{block_gray} = [draw=black, fill=gray!80, thick, minimum size=7mm,
    rectangle, rounded corners, inner sep=3pt, inner ysep=2pt, text centered]

\begin{figure}
\centering
\begin{tikzpicture}

\node[block](A){};

\node[block,
right of=A,
yshift=0cm,
xshift=.5cm,
anchor=center
](B){};
\draw (B.west) edge[-latex, semithick] (A.east);

\node[block,
above of=B,
yshift=.3cm,
xshift=0cm,
anchor=center
](BUP){};
\draw (BUP.west) edge[-latex, semithick] (A.east);

\node[block_gray,
below of=B,
yshift=-.3cm,
xshift=0cm,
anchor=center
](BDOWN){};
\draw (BDOWN.west) edge[-latex, semithick] (A.east);

\node[block_blue,
right of=B,
yshift=0cm,
xshift=.5cm,
anchor=center
](C){};

\draw (C.west) edge[-latex, semithick] (B.east);
\draw (C.west) edge[-latex, semithick, dashed] (BUP.east);

\node[block_gray,
above of=C,
yshift=.3cm,
xshift=0cm,
anchor=center
](CUP){};
\draw (CUP.west) edge[-latex, semithick] (B.east);

\node[block_gray,
below of=C,
yshift=-.3cm,
xshift=0cm,
anchor=center
](CDOWN){};
\draw (CDOWN.west) edge[-latex, semithick] (BDOWN.east);

\node[block,
right of=CDOWN,
yshift=0cm,
xshift=.5cm,
anchor=center
](DDOWN){};
\draw (DDOWN.west) edge[-latex, semithick] (CDOWN.east);
\draw (DDOWN.west) edge[-latex, semithick, dashed] (C.east);

\node[block,
right of=C,
yshift=0cm,
xshift=.5cm,
anchor=center
](D){};
\draw (D.west) edge[-latex, semithick] (C.east);
\draw (D.west) edge[-latex, semithick, dashed] (CUP.east);

\end{tikzpicture}

\caption{An example of a conflict set. The gray blocks constitute the conflict set of the blue block. The dashed arrows are references and the solid arrows are parent references.}
\label{fig:ex2}
\end{figure}

Intuitively, the scheme we propose awards every block some amount of reward $b$ decreased by a penalty $c$ multiplied by the size of the conflict set. The ultimate purpose of the properties we establish is to make sure that rational miners want to minimize the conflict set of the blocks they create, following the protocol as a consequence.

\begin{definition}[Rewards]\label{def:rewards}
A reward scheme $\cali R^{c,b}$ is such that given the main chain ending with a block $A$, each block $B \in \past(A)$ is granted $\cali R_A^{c,b}(B)$ amount of reward:
\begin{equation*}
\cali R_A^{c,b}(B) =
\begin{cases}
0, & \text{if $B \in S_A$ or $D(A,\lca(A,B)) \le 2p$}.\\
b - c | X_A(B) |, & \text{otherwise}.
\end{cases}
\end{equation*}
\end{definition}

We write $\cali R^c$ for $\cali R^{c,b}$ if $b$ is clear from context, or just $\cali R$ if $c$ is clear from context.

In our reward scheme, the reward associated with a given block are decreased linearly with the size of the block's conflict set. We need to ensure that no block reward is negative, otherwise the reward scheme would break down. Lemma \ref{lem:linearconflictset} shows that it is only possible for the conflict set to reach certain size; the probability that the conflict set of a block is bigger than linear in $p$ is negligible. Intuitively, it is because stale blocks cannot be part of a conflict set, and after enough time has passed from broadcasting some block $B$, new blocks either reference $B$ or are stale.

As a consequence, we establish in Corollary \ref{cor:rnonnegative} that the rewards are non-negative.

\begin{lemma}\label{lem:linearconflictset}
Let $x \ge p$ and $B$ be a block. The probability that any honest player adopts a main chain ending with a block $A$ such that $|X_A(B)| > xp$ is $e^{-\Omega(x)}$.
\end{lemma}

The proof is deferred to the appendix.

\begin{corollary}[Rewards Are Non-Negative]\label{cor:rnonnegative}
Let $B$ be a block. The probability that any honest player adopts a main chain ending with a block $A$ such that $\cali R_A^{c,b}(B) < 0$ is $e^{-\Omega(\frac{b}{cp})}$.
\end{corollary}
\begin{proof}
Follows directly from Lemma \ref{lem:linearconflictset}.
\qed \end{proof}

The conflict set of a block is determined based on the main chain. At some point, the reward needs to be determined and stay fixed. Lemma \ref{lem:stableconflictset} shows that if the main chain has grown far enough from block $B$, the new block $A$ appended to the chain will not modify the conflict set of $B$.

\begin{lemma}\label{lem:stableconflictset}
If $D(P(A),\lca(P(A),B)) > 2p$ then $X_A(B) = X_{P(A)}(B)$
\end{lemma}

The proof is deferred to the appendix.

The rewards in Definition \ref{def:rewards} are only assigned as non-zero to blocks $B$ such that $D(A,\lca(A,B)) > 2p$, where $A$ is the block at the end of the main chain. By Corollary \ref{cor:rfinal}, these non-zero rewards are not modified by the blocks extending the main chain and remain fixed.

\begin{corollary}[Rewards Are Final]\label{cor:rfinal}
$$\forall B \in \past(A): \cali R_{P(A)}(B) \ne 0 \implies \cali R_A(B) = \cali R_{P(A)}(B).$$
\end{corollary}
\begin{proof}
$\cali R_A^{c,b}(B)$ is non-zero only if $D(A,\lca(A,B)) > 2p$. The corollary follows from Lemmas \ref{cor:staleparent} and \ref{lem:stableconflictset} and induction.
\qed \end{proof}

The properties we have established so far culminate in Theorem \ref{thm:main}.

\begin{theorem}\label{thm:main}
Deviating from the protocol reduces the adversary's rewards and its proportion of rewards $\cali R^{c,b}$, except with negligible probability.
\end{theorem}

The proof is deferred to the appendix.

\subsection{Nash Equilibria}

Theorem \ref{thm:main} follows from Lemma \ref{lem:freshblock} and hence holds for the same action space as considered in \cite{kiayias2017trees}, i.e. attempting to mine any chosen blocks and withholding or releasing blocks at will. Hence, for this action space, minimizing the conflict set of mined blocks is in the interest of the miner. The adversary is considered as a coordinated minority coalition of players, hence the constants $p,c,b$ can be set such that all players following the protocol constitute a strict, strong Nash equilibrium. In other words, all agents and all minority coalitions of agents strictly prefer to follow the protocol to any alternative strategy.

\begin{corollary}
All players following the protocol constitute a strict, strong Nash equilibrium.
\end{corollary}

However, there exist other Nash equilibria, such as the scenario described in Example \ref{ex:1} in the appendix. The presented equilibrium is based on a player threatening to induce penalties for other players by suffering penalties herself. Intuitively speaking, we suggest all Nash equilibria where some player does not follow the protocol are of this nature, but we do not formalize this concept. However, if the adversary wishes to spend resources solely to influence the behaviour of rational miners, there are always ways to achieve this outside the scope of any reward scheme, such as bribery (see Section \ref{sec:bribery}).

\subsection{Hurting Other Players}

When designing a reward scheme, it might be seen as fair if each honest player is rewarded irrespectively of the strategies of other players. Such fairness principle is enjoyed by the Fruitchains protocol and our reward scheme $\cali R^0$. However, those schemes inevitably trivialize some aspect of the game and leave potential for deviation that goes unpunished. A relaxation of this principle is stated in Corollary \ref{cor:hurting} based on Theorem \ref{thm:main} and its proof.

\begin{corollary}\label{cor:hurting}
Under the reward scheme $\cali R^{c,b}$, by deviating from the protocol the adversary can only reduce the rewards of other players by forfeiting at least the same amount.
\end{corollary}

We observe that the property stated in Corollary \ref{cor:hurting} prevents the existence of selfish mining strategies such as those concerning Bitcoin and other traditional blockchains (see Section \ref{sec:selfish}). Such strategies pose a threat since they enable forfeiting some rewards to penalize other players to an even bigger extent.

\section{Related Work}\label{sec:related}

The model of round-based communication in the setting of blockchain was introduced in \cite{garay2015bitcoin}. This paper formalizes and studies the security of Bitcoin.

\subsection{Selfish Mining}\label{sec:selfish}

Selfish mining is a branch of research studying a type of strategies increasing the proportion of rewards obtained by players in a Bitcoin-like system. Selfish mining exemplifies concerns stemming from the lack of proven incentive compatibility. Selfish mining was first described formally in \cite{eyal2014selfish}, although the idea had been discussed earlier \cite{selfishpost}. Selfish mining strategies have been improved \cite{sapirshtein2016optimal} and generalized \cite{nayak2015stubborn}. Selfish mining is not applicable to our incentive scheme.

\subsection{DAG}

The way we order all blocks for the purpose of processing them was introduced in \cite{lewenberg2015inclusive}. The authors consider an incentive scheme to accompany this modification. Their design relies on altruism, as referring extra blocks has no benefit, other than to creators of referred blocks. Hence, rational miners would never refer them, possibly degenerating the DAG to a blockchain similar to Bitcoin's. Some other shortcomings are discussed by the authors.

The authors of \cite{li2018scaling} contribute an experimental implementation of the directed acyclic graph structure and ordering of \cite{lewenberg2015inclusive}, in particular its advantages with respect to the throughput.

\subsection{Fruitchains}\label{sec:fruitchain}

Fruitchains \cite{pass2017fruitchains} is the work probably the closest related to ours. Fruitchains is a protocol that gives a guarantee that miners are rewarded somewhat proportionally to their mining power. The objective might seem similar to ours, but there are fundamental differences. To achieve fairness, similarly to existing solutions, the Fruitchains protocol requires the majority of miners to cooperate without an incentive. In other words, in order to contribute to the common good of the system, players must put in altruistic work. In contrast, we strive for a protocol such that any miner simply trying to maximize their share or amount of rewards will inadvertently conform to the protocol.

The Fruitchains protocol rewards mining of ``fruits", which are a kind of blocks that do not contribute to the security of the system. The Fruitchains protocol relies on merged-mining \footnote{One of the first mentions of merged-mining as used today is \cite{mergepost}, although the general idea was mentioned as early as \cite{mergemineidea}.} also called 2-for-1 PoW in \cite{garay2015bitcoin}. In addition to fruits, the miners can mine ``normal" blocks (containing the fruits) with minimal extra effort and for no reward. The functioning and security of the system depends only on mining normal blocks according to the protocol.

Miners are asked to reference the fruits of other miners, benefiting others but not themselves, similarly to \cite{lewenberg2015inclusive}. The probability of not doing so having any effect is negligible, since majority of the miners are still assumed to reference said fruits.

The resulting system-wide cooperation guarantees fairness, inevitably removing many game-theoretic aspects from the resulting game. In particular, misbehaviour does not result in any punishment. It is common to analyze blockchain designs with respect to the expected cost of a double-spend attempt. In the case of Fruitchains, while the probability of double-spends being successful is similar to previous designs, the {\it cost} of attempting to double-spend is nullified. As a result, any miner might attempt to double-spend constantly at no cost, which we view as a serious jeopardy to the system.

In the absence of punishments, we also argue that not conforming to the protocol is often simpler. Since transaction fees are shared between miners, including transactions might be seen as pointless altogether. Mining only fruits with dummy, zero-fee transactions, while not including the fruits of others (or not mining for blocks altogether), would relieve the miner of a vast majority of the network communication.

Another game-theoretic issue of the Fruitchains protocol is that while it prescribes sharing of the transaction fees, miners might ask transaction issuers to disguise the fee as an additional transaction output, locking it to a specific miner, potentially benefiting both parties and disrupting the protocol.

As argued in Section \ref{sec:rewards}, the reward scheme $\cali R^0$ is an improvement over Fruitchains in the same vein, achieving the same result while avoiding some of the complications.

In contrast to Fruitchains protocol, the approach of reward schemes $\cali R^{c,b}$ is to employ purely economic forces, clearly incentivizing desired behaviour while making sure that deviations are punished.

\subsection{Bribery}\label{sec:bribery}

Recently, there have been works highlighting the problems of bribery, e.g. \cite{Bonneau16,judmayer2019pay,McCorryHM18}. A bribing attacker might temporarily convince some otherwise honest players (either using threats or incentives) to join the adversary. Consequently, the adversary might gain more than half of the computational power, taking over the system temporarily.

Such bribery might be completely external to the reward scheme itself, for example the adversary might program a smart contract (perhaps in another blockchain) that provably offers rewards to miners that show they deviate from the protocol \cite{judmayer2019pay}. Hence, no permissionless blockchain can be safe against this type of attack.

\section{Conclusions}

Mining is a risky business, as block rewards must pay for hardware investments, energy and other operation costs. At the time of this writing, the Bitcoin mining turnover alone is worth over \$10 billion per year, which is without a doubt a serious market. Miners in this market are professionals, who will make sure that their investments pay off. Yet, many believe that a majority of miners will follow the protocol altruistically, in the best interests of everybody, the ``greater good".

We argue that assuming altruistic miners is not strong enough to be a foundation for a reliable protocol. In this work, we introduced a blockchain incentive scheme such that following the protocol is guaranteed to be the optimal strategy.

We showed that our design is tolerant to miners acting rationally, trying to get the maximum possible rewards, with no consideration for the overall health of the blockchain.

To the best of our knowledge, our design is the first to provably allow for rational mining.  Nakamoto \cite{nakamoto2008bitcoin} needed  ``honest nodes collectively control more CPU power than any cooperating group of attacker nodes". With our design it is possible to turn the word honest into the word rational.

\bibliographystyle{splncs04}
\bibliography{refs}

\section*{Appendix}

\appendix

\section{Discussion of Block Content and Transaction Fees}\label{sec:blockcontents}

Depending on the use of the blockchain, miners can be rewarded for including contents in their blocks in various ways. Typically, a transaction fee is awarded to only one miner that first includes the transaction in a block. As a result, the order of processing blocks is important for determining who collects the fees, as it indicates which block is the first. Problematic incentives are introduced with respect to manipulating the order.

Any particular fee-sharing scheme cannot be enforced, because the fee might be disguised as a regular transaction output paid to the miner directly. This can benefit both the transaction issuer and the miner, incentivizing the behaviour.\footnote{If we disregard this vulnerability, the same fee-sharing approach as employed by the Fruitchains protocol can be applied to our work.}

To be incentive compatible, it is not necessary that the fees are spread proportionally. What we want is that the miners never have an incentive to omit a reference to another block. As all blocks are assumed to eventually be included in the blockchain, it is enough to ensure that sufficiently small changes of the linearized order of the blocks have no effect on the miner rewards. This can be achieved by allowing multiple blocks to claim the same inclusion of contents, and having the fee be shared among the including blocks equally. 

In other words, any player who wishes to include a transaction can do so within a certain window, without an effect on their incentives to reference other blocks. Crucially, sending the fee directly to a miner as a transaction output removes the incentive for other miners to include the transaction, as well as the incentive to manipulate the place of the including block in the order.

The point of such a change would be to separate transaction inclusion from referencing blocks. Transaction inclusion is a complex game in itself, similar to the game studied by \cite{lewenberg2015inclusive}.

\section{Proofs}

\paragraph{\bf Proof of Theorem \ref{thm:honestarenonstale}.}
{\itshape
Let $B$ be an honest block mined on round $r$. With probability $1 - e^{-\Omega(p)}$, after round $r + O(p)$ each honest player $H$ adopts a main chain ending with a block $B_H$ such that $B \in \bar S_{B_H}$.
}
\begin{proof}
Let $\Delta = \floor{\frac{p}{2 (\alpha+\beta) (1 + \epsilon)} - \frac{1}{2}} = O(p)$. By Lemma \ref{lem:freshblock}, with probability $1 - e^{-\Omega(\Delta)}$, on and after round $r$, honest players have adopted main chains containing a block $C$ mined between rounds $r - \Delta$ and $r$ (or the genesis block if $r - \Delta < 1$). Hence $C$ is an ancestor of $B$. By Lemma \ref{lem:freshblock}, let $D$ be the honest block mined between rounds $r+1$ and $r + \Delta + 1$ that honest players adopted in the main chain on and after round $r + \Delta + 1$, again with probability $1 - e^{-\Omega(\Delta)}$. $D$ is honest and mined after round $r$, so $B \in \past(D)$. 

Since $C$ was mined on or after round $r - \Delta$, and $D$ was mined on or before round $r + \Delta + 1$, $D(C,D)$ is at most the number $Y$ of blocks mined between rounds $r - \Delta$ and $r + \Delta + 1$. By the Chernoff bound:
$$
e^{-\frac{\epsilon^2 (2(\alpha+\beta)\Delta)}{3}} \geq \Pr[Y \geq (1 + \epsilon)(\alpha+\beta)(2\Delta + 1)] \geq \Pr[Y \geq p]
$$

Since $C$ is an ancestor of $D$, $C$ is an ancestor of $\lca(B,D)$, and $D(C,D) \geq D(\lca(B,D), D)$. By Algorithm \ref{algo:stale}:
$$
D(C,D) < p \implies B \in \bar S_D.
$$

By union bound, the probability that such $C$ and $D$ exist and that $B \in \bar S_D$ is at least equal
$$
1 - 2e^{-\Omega(\Delta)} - e^{-\frac{\epsilon^2 (2(\alpha+\beta)\Delta)}{3}} = 1 - e^{-\Omega(p)}.
$$
By Corollary \ref{cor:staleparent} and induction, with probability $1 - e^{-\Omega(p)}$, after round $r+\Delta$ all honest players adopt only chains ending with blocks $X$ such that $B \in \bar S_X$.
\qed \end{proof}

\paragraph{\bf Proof of Lemma \ref{lem:linearconflictset}}
{\itshape
Let $x \ge p$ and $B$ be a block. The probability that any honest player adopts a main chain ending with a block $A$ such that $|X_A(B)| > xp$ is $e^{-\Omega(x)}$.
}

\begin{proof}
Let $r$ be the round $B$ was announced. Let $P_i$, $i \in \{1, \dots, 2p\}$ (respectively $F_i$, $i \in \{1, \dots, p\}$), be an honest block mined between rounds $r - \frac{xi}{4} - 1$ and $r - \frac{x(i-1)}{4} - 1$ (resp. $r + \frac{x(i-1)}{4} + 1$ and $r + \frac{xi}{4} + 1$) contained in the main chain of every honest player on and after round $r + \frac{xp}{4} + 1$; by Lemma \ref{lem:freshblock} and union bound such blocks exist with probability $1 - e^{-\Omega(x)}$.

Since $F_1$ is honest, $B \in \past(F_1)$. By Algorithm \ref{algo:stale}, if $P_p \notin \past(B)$, then $B \in S_{F_1}$ and $X_A(B)$ remains undefined for honest players. Otherwise, assume $P_p \in \past(B)$.

Let $Z$ be a block such that $Z \notin \past(B) \land B \notin \past(Z)$. Since $Z \notin \past(B)$, $Z \notin \past(P_p)$. By Algorithm \ref{algo:stale}, either $P_{2p} \in \past(Z)$, or $Z$ becomes stale in the main chains of honest players from round $r - \frac{x(p-1)}{4} - 1$ on. Assume $P_{2p} \in \past(Z)$, and hence $Z$ is mined on or after round $r - \frac{2xp}{4} - 1$.

Since $B \notin \past(Z)$, $F_1 \notin \past(Z)$. Then, either $Z$ is announced before round $r + \frac{xp}{4} + 1$, or by Algorithm \ref{algo:stale}, $Z$ becomes stale in the main chains of honest players afterwards. Assume $Z$ is announced before round $r + \frac{xp}{4} + 1$.

Therefore, $Z \in X_A(B)$ implies that $Z$ is mined between rounds $r - \frac{2xp}{4} - 1$ and $r + \frac{xp}{4} + 1$. Let $Y$ be the number of blocks mined between these rounds. By Chernoff bound:
$$
\Pr[Y \ge xp] \le \Pr[Y \ge \frac{4}{3} (\alpha+\beta) (\frac{3xp}{4} + 2)] = e^{-\Omega(x)}.
$$

Note the bound is appliable to any main chain of an honest player before round $r + \frac{xp}{4} + 1$ as well. The claim follows from the union bound.
\qed \end{proof}

\paragraph{\bf Proof of Lemma \ref{lem:stableconflictset}}
{\itshape
If $D(P(A),\lca(P(A),B)) > 2p$ then $X_A(B) = X_{P(A)}(B)$
}
\begin{proof}
From Definition \ref{def:conflictset}, $X_{P(A)}(B) \subseteq X_A(B)$. Suppose for contradiction $\exists Y: Y \in X_A(B) \setminus X_{P(A)}(B)$. From Definition \ref{def:conflictset}, $B \in \bar S_A$, therefore $B \in \past(P^i(A))$. Hence, $P^i(A) \notin \past(Y)$. Since $Y \notin \past(P(A))$, $D(A,\lca(A,Y)) > p$ and $Y \in S_A$, a contradiction.
\qed \end{proof}

\paragraph{\bf Proof of Theorem \ref{thm:main}}
{\itshape
Deviating from the protocol reduces the adversary's rewards and its proportion of rewards $\cali R^{c,b}$, except with negligible probability.
}
\begin{proof}
Honest blocks are not-stale, except with negligible probability (Theorem \ref{thm:honestarenonstale}). Block rewards are final and non-negative, except with negligible probability (Corollaries \ref{cor:rfinal} and \ref{cor:rnonnegative}). Hence, eventual value of $\cali R^{c,b}(B)$ for honest blocks $B$ depends only on $|X_A(B)|$. Since $Y \in X_A(Z) \iff Z \in X_A(Y)$, by increasing $|X_A(B)|$ of an honest block the adversary can only reduce the rewards of honest players (by $c |X_A(B)|$) if the adversary forfeits the same amount. Since the adversary constitutes a minority, its proportion of rewards decreases as well.

The adversary can also produce stale blocks, forfeiting the otherwise non-negative reward, while not changing the rewards of honest players.

Not referencing some known honest block directly increases $|X_A(B)|$. Withholding a block might only prevent some honest player from referencing it, thus increasing $|X_A(B)|$. Hence, any strategy effectively different from the protocol increases $|X_A(B)|$ of produced blocks, thus decreasing the amount and share of rewards of the adversary.
\qed \end{proof}

\subsection{Example Equilibrium}

\begin{example}\label{ex:1}
Four blockchain players with equal hashing power each adopt the following strategies:
\begin{itemize}
    \item Player 1 and 2: Follow the protocol.
    \item Player 3: Do not broadcast new blocks in the first round, otherwise follow the protocol.
    \item Player 4: If you receive three  blocks of other players (child blocks of the genesis block) at the beginning of the second round, then induce penalties for yourself and other players as much as possible forever. Otherwise follow the protocol.
\end{itemize}
Since Player 3 refrains from broadcasting blocks in the first round, Player 4 can never receive three blocks of other players in the second round, and thus the strategy of Player 4 is identical to following the protocol.

However, Player 3 deviates from the protocol. With some constant probability Players 1 and 2 broadcast a block each in the first round, so if Player 3 broadcast a block in the first round, Player 4 could receive three blocks. Then, this action {\it would} change the behaviour of Player 4 to cause penalties to herself and Player 3. Hence, the strategy profile is a Nash equilibrium.
\end{example}

\end{document}